\documentclass{article}

\usepackage{amsmath}
\usepackage{amssymb}
\usepackage{mathrsfs}
\usepackage{amsthm}
\usepackage{textcomp}
\usepackage{xcolor}
\usepackage{graphicx}
\usepackage{enumerate}
\usepackage{mathtools}
\usepackage{hyperref}
\usepackage[title]{appendix}
\usepackage{subfigure}
\usepackage{bm}
\usepackage[top=2cm, bottom=2cm, left=2cm, right=3cm]{geometry} 
\usepackage{dsfont}
\usepackage{todonotes}
\usepackage{framed}
\usepackage{wrapfig}
\usepackage{lipsum}

\newtheorem{theorem}{Theorem}[section]
\newtheorem{remark}[theorem]{Remark}
\newtheorem{example}[theorem]{Example}
\newtheorem{proposition}[theorem]{Proposition}
\newtheorem{corollary}[theorem]{Corollary}
\newtheorem{lemma}[theorem]{Lemma}
\newtheorem{definition}[theorem]{Definition}

\newcommand{\V}{{\cal V}}
\newcommand{\E}{{\cal E}}
\newcommand{\ELL}{{\cal L}}
\newcommand{\J}{{\cal J}}

\begin{document}

\title{The Kuramoto model on oriented and signed graphs\thanks{This work has been supported by ETH Z\"urich funding and the SNSF AP Energy grant PYAPP2\_154275}}
\author{Robin Delabays \footnotemark[2] \footnotemark[3]
\and Philippe Jacquod \footnote{University of Applied Science of Western Switzerland, CH-1950 Sion, Switzerland, ({\tt robin.delabays@hevs.ch}, {\tt philippe.jacquod@hevs.ch}). }
\and Florian D\"orfler \footnote{Automatic Control Laboratory, Swiss Federal Institute of Technology (ETH) Z\"urich,  Switzerland, ({\tt dorfler@ethz.ch}).}
}


\date{\today}

\maketitle

\begin{abstract}
 Many real-world systems of coupled agents exhibit directed interactions, meaning that the influence of an agent on another is not reciprocal. 
 Furthermore, interactions usually do not have identical amplitude and/or sign. 
 To describe synchronization phenomena in such systems, we use a generalized Kuramoto model with oriented, weighted and signed interactions. 
 Taking a bottom-up approach, we investigate the simplest possible oriented networks, namely acyclic oriented networks and oriented cycles. 
 These two types of networks are fundamental building blocks from which many general oriented networks can be constructed. 
 For acyclic, weighted and signed networks, we are able to completely characterize synchronization properties through necessary and sufficient conditions, which we show are optimal. 
 Additionally, we prove that if it exists, a stable synchronous state is unique. 
 In oriented, weighted and signed cycles with identical natural frequencies, we show that the system globally synchronizes and that the number of stable synchronous states is finite. 
\end{abstract}

\paragraph{Keyword.}
Kuramoto model, synchronization, directed graphs, oriented graphs, weighted graphs, signed graphs.

\paragraph{AMS subject classifications.}
34D06, 37N35

\section{Introduction}
Since its introduction in 1975~\cite{Kur75}, the Kuramoto model has become a standard mathematical model to describe the large 
variety of synchronization phenomena encountered in natural and man-made systems. 
References~\cite{Str00,Ace05,Dor14} give some extensive surveys about it. 
In its initial formulation, the Kuramoto model describes the time evolution of a group of $n$ oscillators, each characterized by a natural frequency $\omega_i\in\mathbb{R}$, 
with identical and symmetric coupling $K/n>0$, 
\begin{align}\label{eq:kuramoto_trad}
 \dot{\theta}_i &= \omega_i - \frac{K}{n}\sum_{j=1}^n\sin(\theta_i-\theta_j)\, , & i &\in \{1,...,n\}\, .
\end{align}
It has been shown~\cite{Erm85,Hem93} that under the assumption that the distribution of natural frequencies has compact support, there exists a critical coupling strength, $K_c>0$, 
such that the oscillators frequency-synchronize, meaning that for all $i,j\in\{1,...,n\}$, 
\begin{align}
 \lim_{t\to\infty}\dot{\theta}_i &= \lim_{t\to\infty}\dot{\theta}_j\, ,
\end{align}
if $K>K_c$. 

The Kuramoto model achieves a good compromise between the simplicity of expression, allowing for an analytical approach, and the complexity of synchronization behaviors described. 
During the last decades, the interest in the Kuramoto model increased in many fields of science and engineering. 
It has been shown that it can be used to describe synchronization phenomena in domains as various as biology~\cite{Buc88,Erm91,Lu16}, physics~\cite{Wie96,Wie98} and engineering~\cite{Dor13}. 
To better represent real-world systems, the Kuramoto model has been generalized in many different ways~\cite{Dor13,Mal13,Chi18}. 
One of the main generalizations is to consider the Kuramoto model with interactions given by an arbitrary graph, 
\begin{align}
 \dot{\theta}_i &= \omega_i - \sum_{j=1}^na_{ij}\sin(\theta_i-\theta_j)\, , & i &\in \{1,...,n\}\, ,
\end{align}
where $a_{ij}$ is the $(i,j)^{\rm th}$ element of the weighted adjacency matrix of the considered graph. 

In the vast majority of the literature on the Kuramoto model, the interactions are considered symmetric and positive, i.e., $a_{ij}=a_{ji}>0$. 
However, in some systems exhibiting synchronizing behaviors, the interaction can be nonsymmetric ($a_{ij}\neq a_{ji}$) 
or even unidirectional ($a_{ij}\neq0\implies a_{ji}=0$)~\cite{Res06,Pal07,Sep08,Mas07}. 
Some systems also exhibit negative couplings, observed in particular in social networks, when some agents adjust their ideas in opposition to some others, 
and in interactions between neurons that can be excitatory ($a_{ij}>0$) or inhibitory ($a_{ij}<0$)~\cite{Pal07,Sep08,Bor03,Sch16}. 
In this manuscript, we focus on the Kuramoto model with interactions in a unique direction, which we refer to as \emph{oriented interactions}, 
to be distinguished from \emph{directed interactions}, which can go in both directions between two vertices. 
We consider general coupling weights, that can be positive or negative, referred to as \emph{signed interactions}. 

Investigations about the Kuramoto model with mixed positive and negative couplings showed that with all-to-all coupling, 
synchronization depends on the ratio between the number of positive and negative couplings~\cite{Hon11,Bur12}. 

In~\cite{Res06}, the authors give some estimation of the critical coupling $K_c$ in the limit of large ($n\gg0$) oriented networks, as well as networks with mixed positive and negative couplings. 
Some conditions for the existence of a frequency-synchronous state in general directed networks are given in~\cite{Ati13a}, 
and~\cite{Ati13b} gives some conditions for the existence of frequency-synchronous states in complete directed, weighted and signed networks. 
The authors of~\cite{Ska16a} show that in directed networks, synchronization is favored if the natural frequencies and the weighted out-degrees 
(called \emph{in-degrees} in their article) of the oscillators are correlated. 

In this manuscript, we are interested in determining conditions for global synchronization of the Kuramoto model with oriented interactions. 
All previously cited references consider either all-to-all coupling or general interaction graphs, without restriction on the type of graphs considered. 
We choose here a bottom-up approach, starting with the simplest oriented networks possible, i.e., oriented acyclic networks and oriented cycles. 
For weighted and signed oriented acyclic networks, we give some explicit necessary and sufficient conditions for almost global synchronization of the system. 
For weighted and signed oriented cycles, with identical natural frequencies, we prove global synchronization and show that the final synchronization frequency can take only discrete values. 

The Kuramoto model on acyclic oriented graphs has been investigated in~\cite{Ha14}, where the authors give some local subdomains of the state space where initial condition are guaranteed to lead 
to synchronization. 
In section~\ref{sec:acyclic}, we almost completely characterize the synchronization properties of the Kuramoto model on weighted and signed oriented acyclic graphs. 
We give some necessary and sufficient conditions on the system's parameters (edge weights and natural frequencies) for the existence of a globally exponentially stable synchronous state. 
Furthermore, this synchronous state is unique. 
We show through examples that our necessary and sufficient conditions cannot be sharpened in general. 
It is remarkable that our results hold regardless of the signs of the couplings in the network. 

On oriented cycles, Rogge and Aeyels~\cite{Rog04} obtained an upper bound on the number of stable synchronous states of the Kuramoto model with cyclic oriented interactions, 
for general natural frequency distribution. 
To the best of our knowledge, the best estimate of the region of synchronization for identical natural frequencies is given in~\cite{Ha12}, 
where the authors explicitly determine a subdomain of the state space where the system always synchronizes. 
They further identify subdomains of the basins of attraction of the various synchronous states. 
In section~\ref{sec:cycles}, we show that, in the case of identical natural frequencies, the system globally synchronizes, again regardless of the sign and magnitude of the couplings. 
We further prove that each stable synchronous state corresponds to a different synchronous frequency, and the the number of such synchronous frequencies is finite. 
We express them as solutions of an equation. 
Furthermore, we observe numerically that the winding number of a synchronous state is correlated to the winding number of the initial conditions. 
The distribution of final winding numbers given the initial winding number is Gaussian. 
Surprisingly, the variance of this distribution is identical for all initial winding numbers. 

In section~\ref{sec:combi}, we illustrate the difficulty to generalize our results to more complex interaction graphs. 
Using oriented acyclic networks and oriented cycles as building blocks, we construct more general oriented graphs where we attempt to extend the results of the previous sections. 
We show that in some cases our results can be directly generalized, but we also identify some cases where the system synchronizes or not depending on the initial conditions of the oscillators. 
The Kuramoto dynamics reveal to be rather complex in general oriented and signed networks. 

\subsection{Definitions and properties}
We consider a generalized Kuramoto model with oriented weighted, and real-valued (not necessarily positive) interactions
\begin{align}\label{eq:main}
 \dot{\theta}_i &= \omega_i -\sum_{j=1}^na_{ij}\sin(\theta_i-\theta_j)\, , & i &\in\{1,...,n\}\, ,
\end{align}
where $\theta_i\in\mathbb{S}^1\simeq\mathbb{R}/(2\pi\mathbb{Z})$ and $\omega_i\in\mathbb{R}$ are the $i^{\rm th}$ oscillator's angle and natural frequency respectively, 
$a_{ij}\in\mathbb{R}$ is the coupling constant between oscillators $i$ and $j$, and $n$ is the number of oscillators. 
The coupling is a priori not symmetric, we allow negative edge weights, and $a_{ij}=0$ means that there is no edge from $i$ to $j$. 
The interaction graph is then directed, signed and weighted. 
The oscillators' angles are aggregated in the \textit{state} $\bm{\theta}\in\mathbb{T}^n$. 
For any $c\in\mathbb{R}$ and $\bm{k}\in\mathbb{Z}^n$, the states $\bm{\theta}$ and $\bm{\theta}+2\pi\cdot\bm{k}+c\cdot(1,...,1)$ are equivalent 
with respect to the dynamics of~\eqref{eq:main}. 
We then consider any state up to a constant shift of all angles and up to addition of integer multiples of $2\pi$. 

A \textit{weighted, signed and directed graph} ${\cal G}=(\V,\E,A)$ is defined by a set of vertices $\V\coloneqq\{1,...,n\}$, a set of edges 
$\E\subset\V\times\V$ which are ordered pairs of vertices, and a weighted adjacency matrix $A$, 
where $a_{ij}\neq0$ if and only if the edge $(i,j)$ from vertex $i$ to vertex $j$ belongs to $\E$. 
In this manuscript, we consider finite ($n<\infty$) \emph{oriented} graphs ($a_{ij}\neq 0 \implies a_{ji}=0$) without self-loops ($a_{ii}=0$). 
In an oriented graph, if there exists an edge from vertex $i$ to vertex $j$ ($a_{ij}\neq 0$), we call $i$ the \textit{parent} of $j$, and $j$ is the \textit{child} of $i$. 
Vertices without children are called \textit{leaders}, and the set of leaders is $\ELL\subset\V$. 
A \textit{path} from vertex $i$ to vertex $j$ is a sequence of edges ${\cal P}_{ij}\coloneqq(\ell_1,...,\ell_L)$, where $i$ is the parent of edge $\ell_1$, $j$ is the child of edge $\ell_L$, 
and for all $p=1,...,L$, $\ell_p$ is an edge of ${\cal G}$ and the parent of $\ell_{p+1}$ is the child of $\ell_p$.  
A path is \textit{simple} if it does not go twice through an edge. 
A \textit{cycle} is a simple path from a vertex to itself. 
If a graph does not contain any cycle, it is \textit{acyclic}. 

\begin{remark}
 We stress out our arbitrary convention of calling \textit{parent} (resp. \textit{child}) the source (resp. target) of an edge. 
 In our definitions, an oscillator is ``watching'' its children, i.e., its dynamics are directly influenced by them. 
 On the contrary, the dynamics of an oscillator is not (directly) influenced by the state of its parents. 
\end{remark}

\subsection{Local order parameter}
For the standard Kuramoto model, with undirected interactions given by \eqref{eq:kuramoto_trad}, the \emph{order parameter}, 
\begin{align}
 re^{i\psi} &\coloneqq \frac{1}{n}\sum_{j=k}^ne^{i\theta_k}\, ,
\end{align}
gives a measure of alignment of the oscillators. 
In particular, the dynamics of each oscillator can be written with respect to the order parameter only~\cite{Dor14}. 
Following the same idea, we define the \textit{local order parameter} of vertex $j$ (a similar definition was given in~\cite{Res06}), 
\begin{align}\label{eq:local_order_param}
 r_je^{i\psi_j} &\coloneqq \sum_{k\neq j}a_{jk}e^{i\theta_k}\, ,
\end{align}
which is equivalently rewritten in real and imaginary parts as 
\begin{align}\label{eq:local_order_param_var}
 r_j\cos(\psi_j) &= \sum_ka_{jk}\cos(\theta_k)\, , &  r_j\sin(\psi_j) &= \sum_ka_{jk}\sin(\theta_k)\, ,
\end{align}
with $\psi_j\in\mathbb{S}^1$ and $r_j\in\mathbb{R}_{\geq 0}$. 
The dynamics~\eqref{eq:main} of each oscillator can then be rewritten 
\begin{align}\label{eq:kuramoto_order}
 \dot{\theta}_i &= \omega_i - r_i\sin(\theta_i-\psi_i)\, , & i &\in \{1,...,n\}\, ,
\end{align}
where we recall that $r_i$ and $\psi_i$ depend on other angles, and are therefore time-dependent. 
We define the \textit{weighted out-degree} of vertex $j$, 
\begin{align}\label{eq:out-degree}
 d_j^{\rm out} &\coloneqq \sum_{k=1}^n|a_{jk}|\, ,
\end{align}
and indexing the children of $j$ from $1$ to $n_j$ such that 
\begin{align}
 |a_{j1}| &\geq |a_{j2}| \geq \ldots \geq |a_{jn_j}|\, ,
\end{align}
we define the \textit{residual outgoing weight} 
\begin{align}\label{eq:row}
 \rho_j &\coloneqq \max\left\{0,|a_{j1}| - \sum_{k=2}^{n_j}|a_{jk}|\right\}\, .
\end{align}

\begin{proposition}\label{prop:order_int}
 For any oriented, weighted and signed graph, the amplitude of the local order parameter $r_j$ belongs to the interval $[\rho_j,d_j^{\rm out}]$. 
 Furthermore, we can construct angle distributions realizing both end values of this interval. 
\end{proposition}

\begin{proof}
 The upper bound is a direct application of the triangle inequality,
 \begin{align}
  r_j &= \left|\sum_ka_{jk}e^{i\theta_k}\right| \leq \sum_k|a_{jk}e^{i\theta_k}| = \sum_k|a_{jk}| = d_j^{\rm out}\, .
 \end{align}
 
 The amplitude of the local order parameter $r_j$ is by definition larger or equal to zero, and the triangle inequality gives also 
 \begin{align}
  r_j &= \left|\sum_ka_{jk}e^{i\theta_k}\right| \geq \left|a_{j1}e^{i\theta_1}\right| - \left|\sum_{k=2}^{n_j}a_{jk}e^{i\theta_k}\right| \geq |a_{j1}| - \sum_{k=2}^{n_j}|a_{jk}|\, .
 \end{align}
 By definition, $\rho_j$ is then a lower bound for $r_j$. 
 This proves the first part of the proposition. 
 
 We prove the second part for nonnegative weights. 
 The argument can be easily adjusted in case some weights are negative. 
 The value $r_j=\rho_j$ is realized when the angles $\theta_k$, $k=2,...,n_j$ are all equal and $\pi$ apart from $\theta_1$: $\theta_2=...=\theta_{n_j}=\theta_1\pm\pi$. 
 When all these angles are identical, $\theta_1=...=\theta_{n_j}$, then $r_j=d^{\rm out}_j$. 
 Thus both ends of the interval can be reached for some angle distributions, both cases are illustrated in Figure~\ref{fig:r_bounds}.
\end{proof}

\begin{figure}
 \centering
 \includegraphics[width=.6\textwidth]{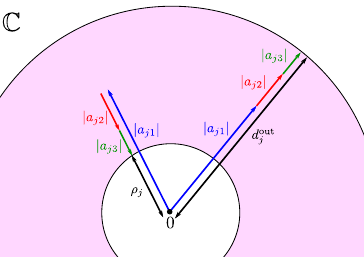}
 \caption{\it The possible complex values of the local order parameter for a vertex with 3 children are in the pink area. 
 The two extreme cases of smallest/largest order parameter amplitude are depicted. 
 When all angles are aligned, the amplitude of the order parameter is $d_j^{\rm out}$ and when all angles are $\pi$ apart from $\theta_1$, the order parameter has amplitude $\rho_j$. }
 \label{fig:r_bounds}
\end{figure}

\begin{remark}
 In the particular case where $j$ has a single out-going edge, then the residual outgoing weight and the weighted out-degree coincide, $\rho_j=d_j^{\rm out}=r_j$. 
 When $j$ has two children, $\rho_j=|a_{j1}|-|a_{j2}|$, which is zero if and only if the weights have the same absolute value. 
 In general, if $j$ has three or more children, $\rho_j$ will be zero if weights are not too different.
\end{remark}

\subsection{Synchronization}
If not specified otherwise, \textit{synchronization} refers to frequency-synchronization, which means $\dot{\theta}_i=\dot{\theta}_j$ for all $i,j$. 
At a synchronous state of the oriented Kuramoto model, all angles rotate at the same frequency and angle differences are constant in time. 
It is necessary for synchronization that all leaders have the same natural frequency and that the deviation from leader frequency is bounded by the weighted out-degree, at each oscillator. 

\begin{proposition}\label{prop:nec}
 Let ${\cal G}$ be an oriented, weighted, and signed graph, with at least one leader. 
 If the dynamical system~(\ref{eq:main}) with interaction graph ${\cal G}$ has a synchronous state, then 
  
  (i) the leaders have identical natural frequency $\omega_i=\omega_{\rm L}$ for all $i\in\ELL$; and 
  
  (ii) for all $i\in\{1,...,n\}$, 
  \begin{align}\label{eq:ubound}
   |\omega_i - \omega_{\rm L}| &\leq d_i^{\rm out}\, .
  \end{align}
\end{proposition}

\begin{proof}
 {\it (i)} By definition, at a synchronous state, all oscillators, and leaders in particular, have the same frequency $\dot{\theta}_i=\omega_{\rm L}$, $i\in\{1,...,n\}$. 
 By definition again, the dynamics of leaders, given by~\eqref{eq:main} reduce to
 \begin{align}\label{eq:dyn_leaders}
  \dot{\theta}_i &= \omega_i\, , & i\in\ELL\, .
 \end{align}
 Therefore the leaders must have the same natural frequency, $\omega_i=\omega_{\rm L}$, $i\in\ELL$. 
 
 {\it (ii)} If $|\omega_i - \omega_{\rm L}|>d_i^{\rm out}$, according to Proposition~\ref{prop:order_int}, 
 the right-hand-side of~\eqref{eq:kuramoto_order} is strictly larger than $\omega_{\rm L}$ in absolute value and the system is then not synchronous. 
 This concludes the proof by contraposition. 
\end{proof}

Without loss of generality, we can consider all angles in a frame rotating at a given frequency. 
Throughout this manuscript, if a network contains leaders and they have identical natural frequency $\omega_{\rm L}$, we apply the change of variable
\begin{align}
 \theta_i(t) &\to \theta_i(t)-\omega_{\rm L}t\, .
\end{align}
After this, leaders have zero natural frequency and a synchronous state at the leader frequency is an equilibrium of~\eqref{eq:main}, meaning that all frequencies are zero. 
We will see that other synchronous states can occur in cyclic interaction graphs.

If for all initial conditions, the system~\eqref{eq:main} converges to a synchronous state, we say that it \emph{globally synchronizes}. 
We say that a synchronous state is \emph{almost globally stable} if almost all initial conditions converge to it. 
If furthermore convergence is exponential, the synchronous state is \emph{almost globally exponentially stable}. 

\section{Oriented acyclic networks}\label{sec:acyclic}
An acyclic oriented graph contains at least one leader. 
Proposition~\ref{prop:nec} gives then necessary conditions for the existence of a synchronous state of~\eqref{eq:main}, and in particular for global synchronization, on such a network. 
We give now a sufficient condition for almost global synchronization in such networks, and prove that there is a unique stable synchronous state. 
Finally, we show that the necessary and sufficient conditions we obtain cannot be improved without further assumptions on the graphs considered. 

\subsection{Equilibria}
Let $\bm{\theta}^*\in\mathbb{T}^n$ be an equilibrium of~\eqref{eq:main}, then it solves 
\begin{align}
 0 &= \omega_i - \sum_{j=1}^na_{ij}\sin(\theta_i^*-\theta_j^*) = \omega_i - r_i^*\sin(\theta_i^*-\psi_i^*)\, , & i &\in\{1,...,n\}\, ,
\end{align}
where $r_i^*$, $\psi_i^*$ are the components of the local order parameter with angles values of $\bm{\theta}^*$. 
Such solutions exist if and only if $\omega_i\leq r_i^*$, for all $i$. 
If $r_i^*>0$, there are then two possible values for $\theta_i^*$ (which coincide if $\omega_i=r_i^*$)  
\begin{align}
 \theta_i^* &\in \{\arcsin(\omega_i/r_i^*)+\psi_i^*, \pi - \arcsin(\omega_i/r_i^*)+\psi_i^*\}\, .
\end{align}
Throughout this manuscript, we take the arcsine to be one-to-one from $[-1,1]$ to $[-\pi/2,\pi/2]$, thus at least one equilibrium satisfies $|\theta_i^*-\psi_i^*|\leq\pi/2$. 

We assess the local stability of $\bm{\theta}^*$ by linearizing the time evolution of a small perturbation of the angles. 
In the dynamical system~\eqref{eq:main} on an acyclic oriented graph, fixed points depend on the initial conditions of the leaders. 
Perturbing the leaders then changes the fixed point, and the analysis is not relevant for the fixed point initially considered. 
To avoid this, we retrict our stability analysis to the dynamics where the leaders are fixed, 
\begin{align}\label{eq:main_restr}
 \dot{\theta}_i &= \omega_i - \sum_{j=1}^na_{ij}\sin(\theta_i-\theta_j)\, , & i &\in\V\setminus\ELL\, , 
\end{align}
with $\theta_i=\theta_i^*$, for $i\in\ELL$. 

\begin{lemma}\label{lem:stab}
 Let $\bm{\theta}^*\in\mathbb{T}^n$ be an equilibrium of~(\ref{eq:main}), where the interaction graph is oriented, acyclic, weighted and signed. 
 The following statements hold: 
 
 (i) If for all $i\in\V\setminus\ELL$, we have $r_i^*>0$ and $|\omega_i|<r_i^*$, 
 then $\bm{\theta}^*$ is locally exponentially stable if and only if $\theta_i^*=\arcsin(\omega_i/r_i^*)+\psi_i^*$, for all $i\in\V\setminus\ELL$. 
 Furthermore, $\bm{\theta}^*$ is the only stable equilibrium with initial conditions of the leaders given by $\theta_i(0)=\theta_i^*$ for $i\in\ELL$. 
 
 (ii) If for at least one $i\in\V\setminus\ELL$ we have $r_i^*>0$, $|\omega_i|<r_i^*$, and $\theta_i^*=\pi-\arcsin(\omega_i/r_i^*)+\psi_i^*$, 
 then $\bm{\theta}^*$ is unstable and its stable manifold of $\bm{\theta}^*$ has zero measure. 
\end{lemma}

\begin{proof}
 {\it (i)} In acyclic oriented networks, the Jacobian matrix ${\cal J}$ of the system~\eqref{eq:main} is lower triangular (up to renumbering the vertices). 
 Its eigenvalues are its diagonal elements 
 \begin{align}
  \J_{ii} &= -\sum_ja_{ij}\cos(\theta_i^*-\theta_j^*) = -r_i^*\cos(\theta_i^*-\psi_i^*)\, .
 \end{align}
 All eigenvalues are then negative if and only if $\theta_i^*=\arcsin(\omega_i/r_i^*)+\psi_i^*$ for all $i\in\V\setminus\ELL$. 
 The equilibrium then satisfies $|\theta_i^*-\psi_i^*|\leq\pi/2$ for $i\in\V\setminus\ELL$. 
 
 Moreover, as the initial conditions of the leaders are fixed, we can construct recursively the angles of the other oscillators. 
 This uniquely defines the stable equilibrium.
 
 {\it (ii)} If for some $i\in\V\setminus\ELL$, we have $\theta_i^*=\pi-\arcsin(\omega_i/r_i^*)+\psi_i^*$, then $\J_{ii}>0$.
 The fixed point $\bm{\theta}^*$ has at least one unstable direction, and its stable manifold is a submanifold of $\mathbb{T}^n$ of dimension at most $n-1$. 
 It has then zero measure~\cite[Proposition 4.1]{Pot09}. 
\end{proof}

\subsection{Sufficient condition}
We showed in Proposition~\ref{prop:nec} that $d_i^{\rm out}>|\omega_i|$ is a necessary condition for synchronization. 
This condition depends only on parameters of the system and is independent of its current state. 
In Lemma~\ref{lem:stab}, we showed that it is sufficient that each nonleader oscillator satisfies $r_i^*>|\omega_i^*|$ to guarantee the existence of a stable equilibrium. 
As the local order parameter depends on initial conditions, this condition is state-dependent. 
A condition that would guarantee synchronization for any initial conditions $\bm{\theta}^\circ\in\mathbb{T}^n$ would be 
\begin{align}\label{eq:real_suff}
 |\omega_i| &< \min_{\bm{\theta}^\circ}r_i^*\, ,
\end{align}
but this last quantity is hard to compute in general. 
However, according to Proposition~\ref{prop:order_int}, we know that $r_i^*\geq\rho_i$ regardless of initial conditions. 
As $\rho_i$ can be computed directly by~\eqref{eq:row}, we give now a state-independent sufficient condition for almost global stability of the unique stable equilibrium. 
Our argument relies on a series of results from~\cite{Mis95}, which we summarize in Appendix~\ref{ap:mischaikow} for completeness. 

\begin{theorem}\label{thm:acyc_suf}
 Let us consider the dynamical system~(\ref{eq:main}), where the interaction graph is oriented, acyclic, weighted, and signed, and assume that 
 for $i\in\{1,...,n\}$, 
 \begin{align}\label{eq:lbound}
  |\omega_i| &< \rho_i\, , ~ \text{if } \rho_i > 0\, , && \text{or} & \omega_i &= 0\, , ~ \text{if } \rho_i = 0\, .
 \end{align}
 Then the system synchronizes exponentially from almost all initial conditions. 
\end{theorem}

\begin{remark}
 The condition in~\eqref{eq:lbound} is always satisfied for the leaders where $\omega_i=\rho_i=0$. 
\end{remark}

\begin{proof}
 We write the initial conditions of the system (at time $t=0$) as $\bm{\theta}^\circ\in\mathbb{T}^n$ and the equilibrium towards which it converges, if it exists, as $\bm{\theta}^*\in\mathbb{T}^n$. 
 If necessary, we renumber the oscillators, such that the adjacency matrix is lower triangular. 
 By recurrence on the oscillators' indices, we prove that all oscillators converge  exponentially to a steady state, i.e., 
 for all $i\in\{1,...,n\}$, $\dot{\theta}_i\to0$ and $\theta_i\to\theta^*_i$ as $t\to\infty$, both exponentially. 
 Note that all variable quantities in this proof depend on initial conditions, but we do not explicitly write this dependence for the sake of readability. 
  
 {\it Base case, $i=1$.}
 By definition, $i=1$ is a leader. 
 According to~\eqref{eq:main}, $\dot{\theta}_1\equiv0$, and it trivially exponentially converges to a steady state, 
 \begin{align}
  \lim_{t\to\infty}\theta_1&=\theta_1^*=\theta_1^\circ\, .
 \end{align}
 
 {\it Induction step.}
 Assume that for all $j\in\{1,...,i-1\}$, $\lim_{t\to\infty}\dot{\theta}_j=0$ and $\lim_{t\to\infty}\theta_j=\theta_j^*$, and convergences are exponential. 
 Let us prove that $i$ converges exponentially to a steady state as well, $\lim_{t\to\infty}\dot{\theta}_i=0$ and $\lim_{t\to\infty}\theta_i=\theta_i^*$. 
 As the adjacency matrix is lower triangular, we can restrict our discussion to the subgraph with vertices $\{1,...,i\}$. 
 
 If $i$ is a leader, the step is done by applying the base case. 
 Assume then that $i$ is not a leader. 
 By assumption, all children $j$ of $i$ converge exponentially, as by renumbering we have $j<i$. 
 Using the local order parameter, we can write~\eqref{eq:main} as the nonautonomous system 
 \begin{align}\label{eq:main_time}
  \dot{\theta}_i &= \omega_i - r_i(t)\sin\left[\theta_i-\psi_i(t)\right] \eqqcolon f(t,\theta_i)\, , 
 \end{align}
 where we write $r_i$ and $\psi_i$ as time-dependent variables because for a given initial state of the system $\bm{\theta}^\circ$, 
 from the point of view of the $i^{\rm th}$ oscillator, these values depend only on time. 
 Children are not influenced by parents, due to the acyclic topology of the graph. 
 The solution $\theta_i$ can be seen as a nonautonomous semiflow (Definition~\ref{def:mischaikow}) for \eqref{eq:main_time}. 
 As all children of $i$ converge to a fixed angle, there exists $r_i^*\in[\rho_i,d_i^{\rm out}]$ and $\psi_i^*\in\mathbb{S}^1$ such that
 \begin{align}\label{eq:convergence}
  \lim_{t\to\infty}r_i(t) &= r_i^*\, , & \lim_{t\to\infty}\psi_i(t) &= \psi_i^*\, .
 \end{align} 
 By definition of the local order parameter, exponential convergence of the children $j$ of $i$ implies that $r_i(t)$ and $\psi_i(t)$ in~\eqref{eq:convergence} converge exponentially. 
 We also define the autonomous system 
 \begin{align}\label{eq:main_limit}
  \dot{\varphi}_i &= \omega_i - r_i^*\sin\left[\varphi_i-\psi_i^*\right] \eqqcolon g(\varphi_i)\, ,
 \end{align} 
 with initial conditions $\bm{\varphi}(0)=\bm{\theta}^\circ$. 
 The solution $\varphi_i$ of \eqref{eq:main_limit} is an autonomous semiflow (Definition~\ref{def:mischaikow}). 
 Observe that~\eqref{eq:main_limit} equals the time-varying dynamics~\eqref{eq:main_time} only if $\psi_i(t)$ and $r_i(t)$ have converged to their steady-state values. 
 We show now that the dynamical systems~\eqref{eq:main_time} and \eqref{eq:main_limit} with the same initial conditions, always converge to the same set of fixed points. 
 
 The $\omega$\textit{-limit set} of $\theta_i$, subject to initial conditions $\bm{\theta}^\circ\in\mathbb{T}^n$, is defined as 
 \begin{align}
  \Lambda_i(\bm{\theta}^\circ) &\coloneqq \left\{\theta\in\mathbb{S}^1\colon\exists~\{t_j\},~j\in\mathbb{N},~{\rm s.t.}~\lim_{j\to\infty}t_j=
  \infty~{\rm and}~\lim_{j\to\infty}\theta_i(t_j)=\theta\right\}\, ,
 \end{align}
 which is the set towards which $\theta_i$ converges under the dynamics of \eqref{eq:main_time}, as $t\to\infty$. 
 
 The functions $f(t,\cdot)$ and $g(\cdot)$ are bounded, continuous and $2\pi$-periodic on $\mathbb{R}$. 
 They are then naturally defined on the compact quotient space $\mathbb{S}^1$, which we parametrize by angles in $(-\pi,\pi]$. 
 Furthermore, as $t\to\infty$, $f$ converges to $g$, and it is a standard result of analysis~\cite[Theorem 7.13]{Rud76} 
 that convergence of a function defined on a compact set is uniform.\footnote{Here \emph{uniform} means that $\forall~\varepsilon>0$, 
 $\exists~T>0$ such that $|f(t,\theta)-g(\theta)|<\varepsilon$, $\forall~t\geq T$ and $\forall~\theta\in\mathbb{S}^1$.} 
 In particular $f(t,\cdot) \to g(\cdot)$ uniformly on any compact subset of $\mathbb{S}^1$, as $t\to\infty$. 
 Proposition~\ref{prop:mischaikow} then implies that $\theta_i$ is asymptotically autonomous with limit semiflow $\varphi_i$ (Definition~\ref{def:mischaikow}).
 Theorem~\ref{thm:mischaikow} implies first that $\theta_i$ converges to $\Lambda_i$, second that $\Lambda_i$ is invariant for $\varphi_i$ (see Definition~\ref{def:mischaikow_inv}), 
 and third that $\Lambda_i$ is chain recurrent for $\varphi_i$ (see Definition~\ref{def:mischaikow_chain_rec}). 
 In the following, we construct the largest invariant and chain recurrent set for the dynamics of equation~\eqref{eq:main_limit}, thereby identifying $\Lambda_i$. 
 
 By Proposition~\ref{prop:order_int}, $r_i^*\geq 0$ and we consider separately the cases $r_i^*>0$ and $r_i^*=0$ . 
 If $r_i^*>0$, the largest invariant and chain recurrent set of the dynamics~\eqref{eq:main_limit} is composed of two points, 
 \begin{align}
  {\cal I}_i &\coloneqq \left\{\arcsin(\omega_i/r_i^*)+\psi_i^*,\pi-\arcsin(\omega_i/r_i^*)+\psi_i^*\right\}\, ,
 \end{align}
 which are well-defined and distinct as we assumed $0\leq|\omega_i|<r_i^*$. 
 This identifies $\Lambda_i$ as the set of equilibria of~\eqref{eq:main_limit}. 
 According to Lemma~\ref{lem:stab}, almost all initial conditions converge to an equilibrium where $\theta_i^*=\arcsin(\omega_i/r_i^*)+\psi_i^*$. 
 Furthermore, applying point (i) of Lemma~\ref{lem:stab} to the subnetwork composed of the vertices $\{1,...,i\}$, the equilibrium is locally exponentially stable, 
 implying exponential convergence of $\theta_i$ to $\theta_i^*$. 
 
 If $r_i^*=0$, then by assumption, $\omega_i=0$. 
 The right-hand-side of~\eqref{eq:main_limit} is then zero for all time, and the right-hand-side of~\eqref{eq:main_time} vanishes as $t\to\infty$.  
 The largest invariant chain recurrent set of the dynamics~\eqref{eq:main_limit} is the whole circle, that is $\Lambda_i=\mathbb{S}^1$. 
 According to the discussion below equation~\eqref{eq:convergence}, convergence of $r_i(t)$ to zero is exponential, i.e., there exists $T>0$ such that for all $t>T$, 
 \begin{align}
  0 &\leq r_i(t) \leq ce^{-\mu t}\, ,
 \end{align}
 with $c$, $\mu > 0$, which implies 
 \begin{align}
  |\dot{\theta}_i| &\leq ce^{-\mu t}\, .
 \end{align}
 We verify that $\theta_i$ converges, 
 \begin{align}
  \lim_{t\to\infty}|\theta_i| &\leq \lim_{t\to\infty} \left[\left|\theta_i(T)\right| + c\int_T^te^{-\mu s}ds \right] = \left|\theta_i(T)\right| + c\mu^{-1}e^{-\mu T} < \infty\, .
 \end{align}
 The limit then exists and we denote it by $\theta_i^*$. 
 Let us show that $\theta_i\to\theta_i^*$ exponentially. 
 For all $t>T$, 
 \begin{align}
  \left|\theta_i(t) - \theta_i^*\right| &\leq \int_t^\infty|\dot{\theta}_i(s)|ds \leq c\mu^{-1}e^{-\mu t}\, .
 \end{align}
 Convergence is then exponential. 
 
 We have shown that convergence is locally exponential for all equilibria $\bm{\theta}^*$ satisfying either $r_i^*>0,\, ~ |\theta_i^*-\psi_i^*| < \pi/2$ or $r_i^* = 0$, for $i\in\{1,...,n\}$. 
 According to Lemma~\ref{lem:stab}, the basin of attraction of all other equilibria have zero measure. 
 Almost all initial conditions then converge to such an equilibrium, and the rate of convergence is exponential since the linearized system converges exponentially~\cite[Theorem 4.15]{Kha02}. 
\end{proof}

If an oscillator $i$ has a single child $j$, then $\rho_i=d_i^{\rm out}=|a_{ij}|$ and, by Theorem~\ref{thm:acyc_suf} and Proposition~\ref{prop:nec}, 
equation~\eqref{eq:lbound} for vertex $i$ is a necessary and sufficient condition for almost global  {exponential} synchronization. 

\begin{corollary}\label{cor:1_out}
 If all vertices have at most one out-going edge, then the system synchronizes almost globally, exponentially fast, if and only if  the natural frequencies of all nonleader oscillators $i$, 
 satisfy $|\omega_i|<|a_{ij}|$, where $j$ is $i$'s unique child. 
\end{corollary}

The case of identical natural frequencies ($\omega_i\equiv0$) is a particular case of Theorem \ref{thm:acyc_suf}.

\begin{corollary}\label{cor:id_freq}
 Consider the dynamical system~(\ref{eq:main}) with identical natural frequencies ($\omega_i\equiv0$) 
 where the interaction graph is oriented, acyclic, weighted, and signed. 
 Then the system synchronizes almost globally, exponentially fast. 
\end{corollary}

\subsection{Tightness of the conditions}\label{sec:tight}
We now give examples of three simple acyclic graphs to illustrate how tight our necessary and sufficient conditions~\eqref{eq:ubound} and \eqref{eq:lbound} are. 
The amplitude of the local order parameter $r_i^*$ depends only on the initial conditions of the leaders $\bm{\theta}^\circ$, and we know that 
\begin{align}\label{eq:min_r}
 \min_{\bm{\theta}^\circ}r_i^* &\geq \rho_i\, .
\end{align}
Depending on the graph topology, the inequality~\eqref{eq:min_r} can be tight or not. 
In the network considered in Example~\ref{ex:tight}, there is equality in~\eqref{eq:min_r}, whereas in Example~\ref{ex:ntight} it is a strict inequality. 

\begin{figure}
 \centering
 \includegraphics[width=.95\textwidth]{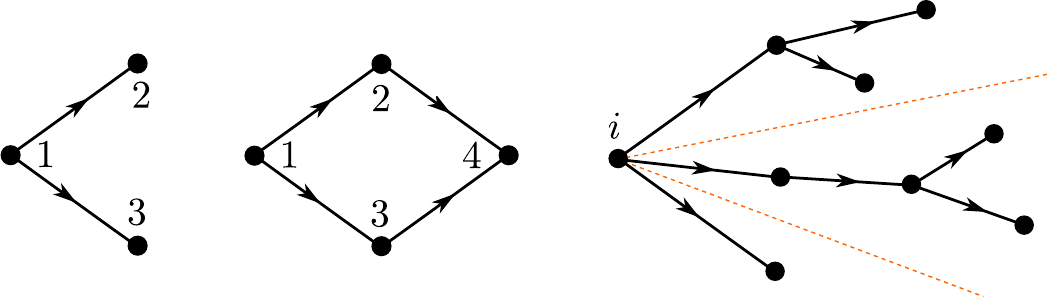}
 \caption{\it Left: Example of an acyclic network where, depending on initial conditions, $r_1^*$ can take any value in the interval $[\rho_i,d_i^{\rm out}]$. 
  Network parameters are $\omega_1=\omega_2=\omega_3=0$ and $(a_{12},a_{13}) = (1,2)$. 
  This gives $\rho_1=1$ and $d_1^{\rm out}=3$. 
  Center: Example of an acyclic network where there exists $\delta>0$ such that $\rho_1+\delta<r_1^*<d_1^{\rm out}-\delta$, for any initial conditions. 
  Network parameters are $(\omega_1,\omega_2,\omega_3,\omega_4)=(0,\sqrt{3},-2,0)$ and $(a_{12},a_{13},a_{24},a_{34})=(1,\sqrt{3},2,4)$. 
  This gives $\rho_1=\sqrt{3}-1$ and $d_1^{\rm out}=\sqrt{3}+1$, while $r_1^*\in\{1,2\}$. 
  Right: Example of an oriented graph where all edges going out of vertex $i$ belong to independent subgraphs, separated by the orange dashed lines. 
  In this case, for any $r\in[\rho_i,d_i^{\rm out}]$, there exists initial conditions such that $r_i^*=r$. }
 \label{fig:example}
\end{figure}

\begin{example}\label{ex:tight}
 Consider the weighted network in the left panel of Figure~\ref{fig:example} (parameters are given in the caption). 
 In this case, $\rho_1=1$ and $d_1^{\rm out}=3$. 
 Consider initial conditions such that $\theta_2^\circ-\theta_3^\circ=\phi$.
 For $\phi=0$, we obtained $r_1^*=d_1^{\rm out}=3$ and for $\phi=\pi$ we have $r_1^*=\rho_1$. 
 Both bounds~\eqref{eq:lbound} and \eqref{eq:ubound} are reached for some initial condition and while $\phi$ is ranging from $0$ to $\pi$, 
 all values of the interval $[\rho_1,d_1^{\rm out}]$ are obtained. 
\end{example}

This example satisfies the hypotheses of Corollaries~\ref{cor:1_out} and \ref{cor:id_freq}.
Both bounds are tight and valid depending on initial conditions, i.e., our bounds are as good as they get. 
For some networks, the sufficient condition of Theorem~\ref{thm:acyc_suf} is as well necessary for global synchronization. 
It also means that for these networks, equation~\eqref{eq:ubound} is a sufficient condition for the existence of an equilibrium. 

\begin{example}\label{ex:ntight}
 Consider the weighted network in the center panel of Figure~\ref{fig:example} (parameters are given in the caption). 
 Here $\rho_1=\sqrt{3}-1$ and $d_1^{\rm out}=\sqrt{3}+1$. 
 According to Theorem~\ref{thm:acyc_suf}, the system almost globally exponentially synchronizes for $|\omega_1|<\rho_1$. 
 We can compute that for all initial conditions, either $r_1^*=1$ or $r_1^*=2$. 
 Then, the system also globally synchronizes for $\rho_1 < |\omega_1| < 1$. 
\end{example}

In this example, the bounds are not sharp: synchronization is almost global for some natural frequencies violating~\eqref{eq:lbound}; 
and impossible for some natural frequencies violating~\eqref{eq:ubound}. 
The system globally synchronizes for $\rho_1 < |\omega_1| \leq \min_{\bm{\theta}^\circ}r_1^*$. 

Generally, if each edge going out of $i$ is connected to a separate subgraph, as in the right panel of Figure~\ref{fig:example}, 
then for any $r\in[\rho_i,d_i^{\rm out}]$, there exist system states such that $r_i^*=r$. 
Otherwise, if two edges going out of $i$ are connected to subgraphs sharing an edge, the range of the local order parameter $r_i^*$ is strictly smaller than the interval $[\rho_i,d_i^{\rm out}]$. 

\begin{example}\label{ex:c_ex}
 Consider the network on the left panel of Figure~\ref{fig:c_ex}, where all coupling are identical. 
 We simulate~(\ref{eq:main}) on this network, with natural frequencies $\bm{\omega}$ given by 
 \begin{align}\label{eq:nat_freq}
  \omega_i &= \left\{
   \begin{array}{ll}
    \omega_0\, , & \text{if } 1\leq i \leq 7\, ,\\
    0\, , & \text{if } 8\leq i \leq 15\, ,
  \end{array}
  \right.
 \end{align}
 and initial conditions satisfying $\theta_i(0)=0$ for $i\in\{8,...,15\}=\ELL$. 
 The center panel of Figure~\ref{fig:c_ex} shows the time evolution of the angle error $\delta\theta_i=|\theta_i-\theta_i^*|$ for $i\in\V\setminus\ELL$, when $\omega_0=1.9<r_i^*$. 
 Convergence is exponential (the y-axis is in log-scale) as predicted by Theorem~\ref{thm:acyc_suf}.
 The right panel of Figure~\ref{fig:c_ex} shows the same curves (dashed line) together with the time evolution of $\delta\theta_i$ when $\omega_0=2=r_i^*$ (plain lines). 
 We clearly see that convergence is slower that exponential (here the axis are both in log-scale). 
 This shows that condition~\eqref{eq:lbound} cannot be relaxed to a nonstrict inequality to guarantee exponential convergence.
\end{example}

\begin{figure}
 \centering
 \includegraphics[width=.95\textwidth]{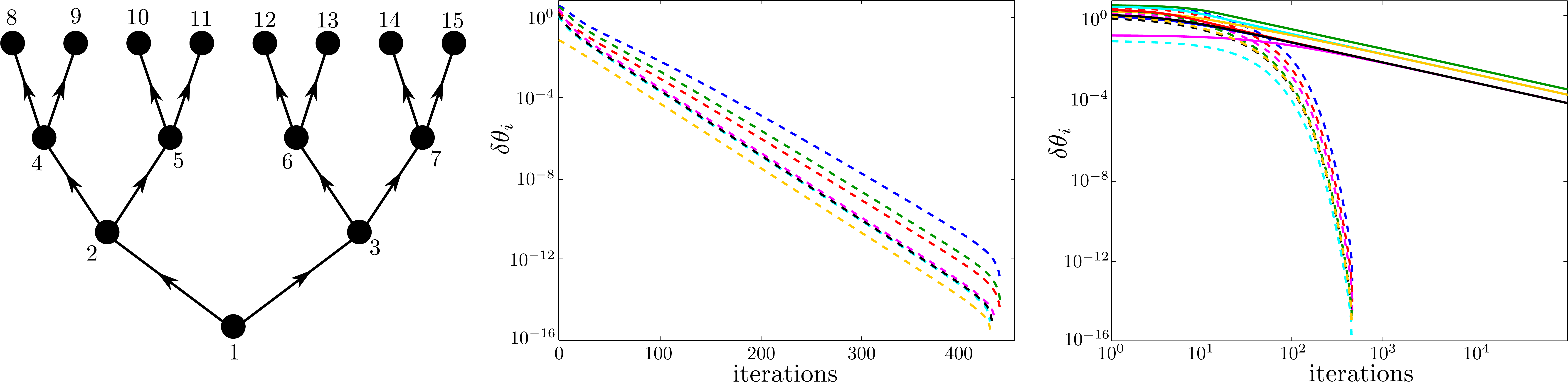}
 \caption{\it Left: Illustration of the network used in Example~\ref{ex:c_ex}. 
 Center: Time evolution of the angle errors $\delta\theta_i$ according to dynamics~(\ref{eq:main}) with natural frequencies given by equation~(\ref{eq:nat_freq}), with $\omega_0=1.9$. 
 Here $\omega_i<r_i^*$ for $i\in\V\setminus\ELL$ and we see exponential convergence (the y-axis is in log-scale). 
 Right: Same curves as in the center panel (dashed lines) together with the time evolution of the angle errors $\delta\theta_i$ according to dynamics~(\ref{eq:main}) 
 with natural frequencies given by equation~(\ref{eq:nat_freq}), with $\omega_0=2$ (plain lines). 
 Here $\omega_i=r_i^*$ for $i\in\V\setminus\ELL$ and we see that convergence is slower than exponential convergence (both axis are in log-scale).}
 \label{fig:c_ex}
\end{figure}

\begin{remark}
 We stress that all results obtained in this section apply to any magnitude and sign of the coupling strengths. 
 Our results can hardly be more general. 
\end{remark}

\section{Oriented cycles with identical frequencies}\label{sec:cycles}
Throughout section~\ref{sec:acyclic}, we completely characterized synchronization in oriented acyclic networks. 
We turn now to the simplest cyclic oriented graph, namely the oriented cycle. 
In this case, we are able to prove global synchronization for identical natural frequencies. 
Again, all the following results hold for any nonzero coupling constants. 

\begin{wrapfigure}{r}{5cm}
 \centering
 \includegraphics[width=.3\textwidth]{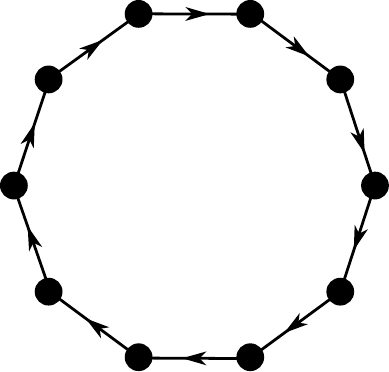}
 \caption{\it Oriented cycle of length $n=10$. 
 In section~\ref{sec:cycles}, we show global synchronization for these graphs.}
 \label{fig:dir_cycle}
\end{wrapfigure}
We consider oriented cycles of length $n$, as in Figure~\ref{fig:dir_cycle}, with arbitrary edge weights. 
By appropriately indexing all vertices, equation~\eqref{eq:main} is then
\begin{align}\label{eq:main_cycle}
 \dot{\theta}_i &= \omega_i - a_{i,i+1}\sin(\theta_i-\theta_{i+1})\, , & i &\in \{1,...,n\}\, ,
\end{align}
where indices are taken modulo $n$, and $a_{i,i+1}\in\mathbb{R}\setminus\{0\}$.

With identical natural frequencies, oriented cycles are very similar to undirected cycles. 
At a synchronous state, both cases admit the same angle differences (center panels of Figure~\ref{fig:dir_vs_undir}). 
The difference being that in the undirected cycle, the oscillators synchronize to the mean natural frequency~\cite[sect.~3.1]{Dor14}, 
which is zero (top right panel of Figure~\ref{fig:dir_vs_undir}). 
At a synchronous state, angles are then constant (top left panel of Figure~\ref{fig:dir_vs_undir}), i.e., a synchronous state is an equilibrium. 
In an oriented cycle, the \emph{synchronization frequency}, which is the frequency of all oscillators at a synchronous state and is denoted by $\omega_{\rm s}$, 
depends on the synchronous state (bottom right panel of Figure~\ref{fig:dir_vs_undir}). 
In this case, oscillators rotate at the same, nonzero, frequency (bottom left panel of Figure~\ref{fig:dir_vs_undir}). 
Stability of synchronous states is similar in oriented~\cite{Rog04} and undirected~\cite{Del16} cycles. 
\begin{figure}
 \centering
 \includegraphics[width=.95\textwidth]{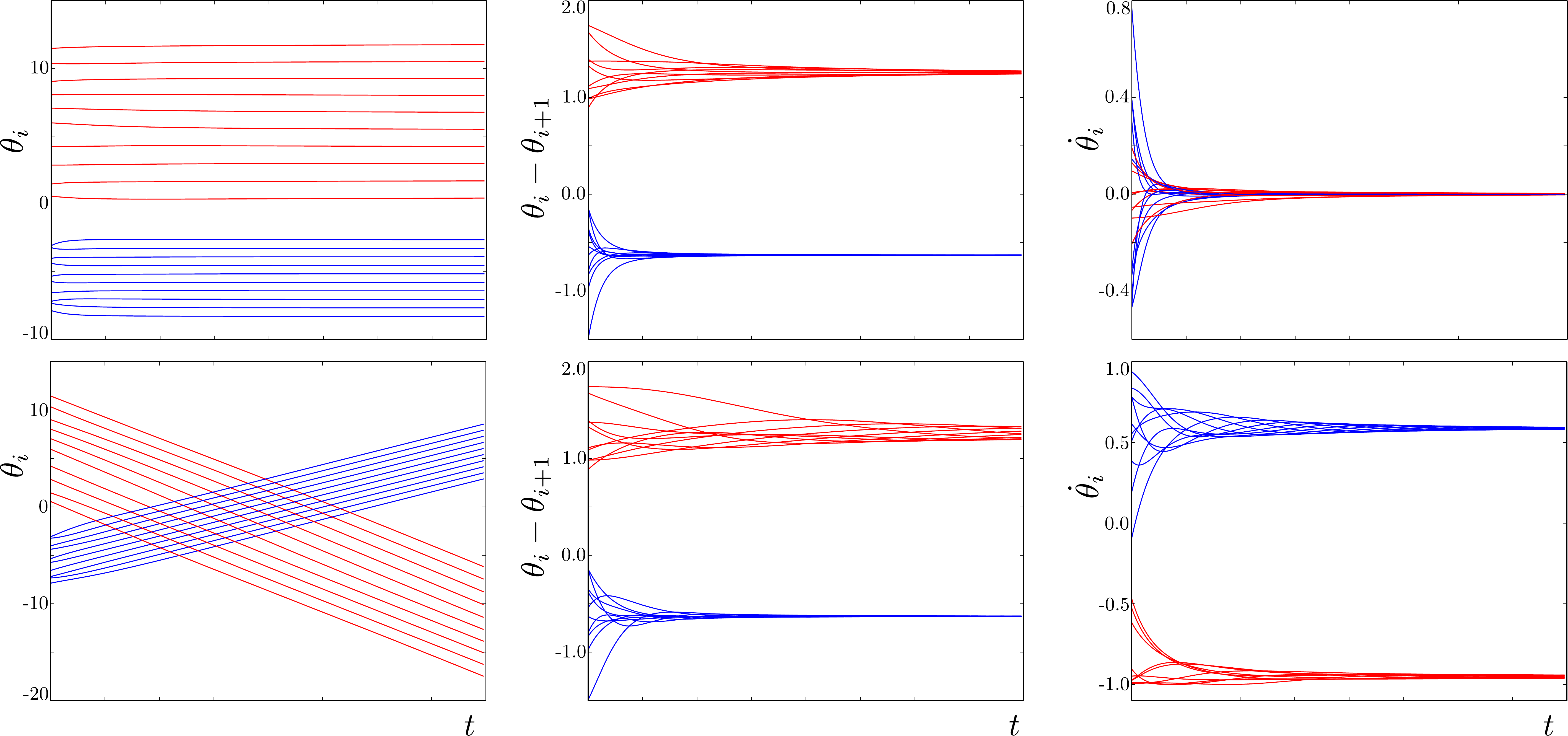}
 \caption{\it Time evolution of angles (left), angle differences (center) and frequencies (right) according to the Kuramoto dynamics (\ref{eq:main}) with undirected (top) 
 and oriented (bottom) cyclic interactions. 
 The cycle is of length $n=10$ with identical natural frequencies and identical couplings. 
 Blue and red lines correspond to two different initial conditions, leading to two different synchronous states. 
 In the undirected case, angles converge to a fixed value (top left panel), which is not the case in the oriented case (bottom left). 
 Angle differences converge to the same value in both undirected and oriented cycles (top and bottom center panels). 
 At the synchronous state, the frequencies are always zero in the undirected cycle (top right panel), and possibly nonzero and different in the oriented case (bottom right panel). }
 \label{fig:dir_vs_undir}
\end{figure}

We will use the following short-hand notations:
\begin{align}\label{eq:shorthand}
 s_i &\coloneqq \sin(\theta_i) & c_i &\coloneqq \cos(\theta_i) &
 s_{ij} &\coloneqq \sin(\theta_i-\theta_j) & c_{ij} &\coloneqq \cos(\theta_i-\theta_j)\, ,
\end{align}
and $\Delta_{ij}$ denotes the angle difference $\theta_i-\theta_j$ taken modulo $2\pi$ in the interval $(-\pi,\pi]$. 

\subsection{Global synchronization}
In the case of identical frequencies ($\omega_i\equiv 0$), we are able to show global synchronization for oriented cycles. 

\begin{theorem}\label{thm:cyc}
 Consider the dynamical system~(\ref{eq:main_cycle}) with identical natural frequencies. 
 This dynamical system globally converges to a synchronous state. 
\end{theorem}

\begin{proof}
 The function
 \begin{align}\label{eq:lasalle}
 \begin{split}
  V\colon \mathbb{T}^n &\longrightarrow \mathbb{R}\\
  \bm{\theta} &\longmapsto 2\cdot\sum_{i,j}a_{ij}\left[1-\cos(\theta_i-\theta_j)\right]\, ,
 \end{split}
 \end{align}
 is a LaSalle function~\cite{Las61} for the system under consideration.
 The function $V$ is bounded, and its time derivative is given by 
 \begin{align}\label{eq:Vdot}
  \dot{V} &= -2\cdot\sum_{i=1}^n \left(a_{i,i+1}^2s_{i,i+1}^2 - a_{i-1,i}a_{i,i+1}s_{i-1,1}s_{i,i+1} \right)\\
  &= -\sum_{i=1}^n \left(a_{i-1,i}s_{i-1,i} - a_{i,i+1}s_{i,i+1}\right)^2 = -\sum_{i=1}^n \left(\dot{\theta}_{i-1} - \dot{\theta}_i\right)^2 \leq 0\, .
 \end{align}
 The system being defined on a compact manifold $\mathbb{T}^n$, LaSalle's invariance principle~\cite{Las61} implies that it converges to a synchronous state, 
 because $V$ is lower bounded and decreasing, and its time derivative vanishes only at synchronous states. 
\end{proof}

The frequency of the oscillators at a synchronous state is bounded by the smallest coupling, 
\begin{align}
 a_{\min} &\coloneqq \min_i|a_{i,i+1}|\, .
\end{align}

\begin{proposition}\label{prop:bound_sync_freq}
 The synchronization frequency $\omega_{\rm s}$ belongs to the interval $[-a_{\min},a_{\min}]$. 
\end{proposition}

\begin{proof}
 The frequency $\omega_{\rm s}$ has to satisfy
 \begin{align}\label{eq:sync_state}
  \omega_{\rm s} &= -a_{i,i+1}s_{i,i+1}\, ,
 \end{align}
 for all $i\in\{1,...,n\}$, which implies 
 \begin{align}\label{eq:bounds}
  \omega_{\rm s} &\in \bigcap_i\left[-|a_{i,i+1}|,|a_{i,i+1}|\right]=[-a_{\min},a_{\min}]\, .
 \end{align}
\end{proof}

\begin{remark}
 In general the set of synchronization frequencies can be continuous. 
 We present here an example inspired from \cite[sect. VI]{Man17}. 
 Consider an oriented cycle whose length $n$ is a multiple of $4$, with identical positive edge weights. 
 Then the state $\bm{\theta}^{\alpha}\in\mathbb{T}^n$, defined by 
 \begin{align}
  \theta_i^{\alpha} &\coloneqq \left\{
  \begin{array}{ll}
   -(i-1)\pi/2\, , & \text{if }i\text{ is odd,} \\
   -(i-2)\pi/2 - \alpha\, , & \text{if }i\text{ is even,}
  \end{array}
  \right.
 \end{align}
 with $\alpha\in(0,\pi/2)$, has angle differences 
 \begin{align}
  \Delta_{i,i+1}^\alpha &= \left\{
  \begin{array}{ll}
   \alpha\, , & \text{if }i\text{ is odd,} \\
   \pi-\alpha\, , & \text{if }i\text{ is even.}
  \end{array}
  \right.
 \end{align}
 It is then an unstable synchronous state of~\eqref{eq:main_cycle}, regardless of the value of $\alpha$. 
 The continuum of values of $\alpha$ implies a continuum of synchronous states.
\end{remark}

Restricting ourselves to stable synchronous states, we show below that their number is discrete.

\subsection{Stability}\label{sec:stab_cyc}
A full stability analysis of the identical frequency case has been done in \cite{Rog04} for positive edge weights. 
We extend it to general nonzero edge weights. 

\begin{lemma}\label{lem:stability}
 A synchronous state $\bm{\theta}^*\in\mathbb{T}^n$ of (\ref{eq:main_cycle}) with identical natural frequencies is locally exponentially stable if and only if 
 \begin{align}\label{eq:stab_condi}
  a_{i,i+1}\cos(\theta_i^*-\theta_{i+1}^*) &> 0\, , & i &\in \{1,...,n\}\, .
 \end{align}
 Furthermore, at such a synchronous state with frequency $\omega_{\rm s}$, the angle differences are given by 
 \begin{align}\label{eq:angle_diffs}
 \Delta_{i,i+1}^{\omega_{\rm s}} &= \left\{
 \begin{array}{ll}
  -\arcsin(\omega_{\rm s}/a_{i,i+1})\, , & \text{if }a_{i,i+1} > 0\, , \\
  \pi + \arcsin(\omega_{\rm s}/a_{i,i+1})\, , & \text{if }a_{i,i+1} < 0\, ,
  \end{array}
 \right.
 & i &\in \{1,...,n\}\, . 
\end{align}
\end{lemma}

\begin{proof}
 A synchronous state is locally exponentially stable if and only if the Jacobian matrix 
 \begin{align}\label{eq:stab_matrix}
  \J(\bm{\theta}) &= 
  \begin{pmatrix}
   -a_{12}c_{12} & a_{12}c_{12} & & & \\
   & -a_{23}c_{23} & a_{23}c_{23} & & \\
   & & \ddots & \ddots & \\
   & & & \ddots & a_{n-1,n}c_{n-1,n} \\
   a_{n,1}c_{n,1} & & & & -a_{n,1}c_{n,1}
  \end{pmatrix}\, ,
 \end{align}
 is Hurwitz. 
 Due to invariance under rotation of all angles, at least one eigenvalue is zero, but this has no influence on the stability of the synchronous state. 
 According to Gershgorin's Circle Theorem~\cite{Hor86}, the eigenvalues of $\J$ are nonpositive if and only if 
 \begin{align}
  a_{i,i+1}\cos(\theta_i^*-\theta_{i+1}^*) &\geq 0\, , & i &\in \{1,...,n\}\, .
 \end{align}
 
 If $\cos(\theta_i^*-\theta_{i+1}^*)=0$ for some $i\in\{1,...,n\}$, then the first-order term in the Taylor series of $\dot{\theta}_i$ is zero. 
 As the second-order term is nonzero, according to the discussion in \cite[sect. 1.2]{Kha91}, the synchronous state is unstable. 
 We conclude that a synchronous state is locally asymptotically stable if and only if
 \begin{align}
  a_{i,i+1}\cos(\theta_i^*-\theta_{i+1}^*) &> 0\, , & i &\in \{1,...,n\}\, .
 \end{align}
 As in this case the Jacobian is Hurwitz, $\bm{\theta}^*$ is locally exponentially stable, which proves the first part of the proposition. 

 For a synchronous state to be stable, the cosine of the angle difference $c_{i,i+1}$ must have the same sign as its edge weight $a_{i,i+1}$. 
 We see that 
 \begin{align}
  \cos[-\arcsin(\omega_{\rm s}/a_{i,i+1})] &\geq 0\, , & \cos[\pi+\arcsin(\omega_{\rm s}/a_{i,i+1})] &\leq 0\, ,
 \end{align}
 where $\omega_{\rm s}$ is the synchronization frequency, and conclude that the angle differences in a stable synchronous state are given by \eqref{eq:angle_diffs}. 
 This concludes the proof. 
\end{proof}

\begin{remark}
 All synchronous states, not only the stable ones, can be parametrized similarly as in~\eqref{eq:angle_diffs}. 
 At a synchronous state, the angle differences satisfy
 \begin{align}
  \Delta_{i,i+1}^{\omega_{\rm s}}\in\left\{-\arcsin(\omega_{\rm s}/a_{i,i+1}),\pi+\arcsin(\omega_{\rm s}/a_{i,i+1})\right\}\, .
 \end{align}
\end{remark}

\subsection{Stable synchronization frequencies}
We now characterize the possible synchronization frequencies corresponding to stable synchronous states, which we call \textit{stable synchronization frequencies}. 
We define $n^-$ to be the number of edges with negative weight and 
\begin{align}
 \sigma &\coloneqq \frac{1}{2\pi}\sum_{i=1}^n\arcsin\left(a_{\min}/|a_{i,i+1}|\right)\, .
\end{align}

\begin{proposition}\label{prop:freq}
 Consider the dynamical system (\ref{eq:main_cycle}) with identical natural frequencies. 
 The number of possible stable synchronization frequencies is discrete and is given by the number of integers in the interval $(n^-/2-\sigma,n^-/2+\sigma)$. 
 In particular, for each stable synchronization frequency $\omega_{\rm s}$, there exists an integer $q\in(n^-/2-\sigma,n^-/2+\sigma)\cap\mathbb{Z}$ such that $\omega_{\rm s}$ solves 
 \begin{align}\label{eq:implicit_freq_stab}
  \sum_{i=1}^n\arcsin\left(\omega_{\rm s}/|a_{i,i+1}|\right) &= (n^--2q)\pi\, .
 \end{align}
\end{proposition}

\begin{proof}
 Let $\omega_{\rm s}\in\mathbb{R}$ be the frequency at the stable synchronous state $\bm{\theta}^{\omega_{\rm s}}\in\mathbb{T}^n$, 
 and let us define $\V^+$ (resp. $\V^-$) the set of vertices whose out-going edge has positive (resp. negative) weight. 
 
 Following the stability analysis of section~\ref{sec:stab_cyc}, the angle differences of the synchronous state are given by~\eqref{eq:angle_diffs}. 
 Furthermore, the sum of angle differences around the cycle has to be an integer multiple of $2\pi$, called the \textit{winding number} $q\in\mathbb{Z}$~\cite{Erm85,Jan03}, 
 \begin{align}\label{eq:winding}
  \sum_{i=1}^n\Delta_{i,i+1} &= 2\pi q\, .
 \end{align}
 
 \begin{remark}
  Note that the winding number is well-defined on any cycle and for any state of the system. 
  Thus at any time of the dynamics, we can compute the winding number which is always an integer. 
 \end{remark}

 As angle differences $\Delta_{i,i+1}^{\omega_{\rm s}}\in(-\pi,\pi]$, the winding number is bounded as $-n/2 \leq q \leq n/2$. 
 Replacing~\eqref{eq:angle_diffs} in \eqref{eq:winding} gives
 \begin{align}
  -\sum_{i\in\V^+}\arcsin(\omega_{\rm s}/a_{i,i+1}) + \sum_{i\in\V^-}\pi + \arcsin(\omega_{\rm s}/a_{i,i+1}) &= 2\pi q \\
  -\sum_{i=1}^n\arcsin(\omega_{\rm s}/|a_{i,i+1}|) &= 2\pi q - \pi n^-\, ,\label{eq:ifs_2}
 \end{align}
 which is~\eqref{eq:implicit_freq_stab}. 
 The left-hand-side of~\eqref{eq:ifs_2} is monotonically decreasing in $\omega_{\rm s}$. 
 For a given winding number $q$, there is then at most one value of $\omega_{\rm s}$ satisfying~\eqref{eq:ifs_2}. 
 As the number of possible values for $q$ is finite, the number of possible synchronization frequencies is finite as well. 
 
 By boundedness of $\omega_{\rm s}$ (Proposition~\ref{prop:bound_sync_freq}) and monotonicity of \eqref{eq:ifs_2}, we can compute the number of possible winding numbers, 
 which is the number of possible stable synchronization frequencies. 
 The lower bound $q_{\min}$ (resp. upper bound $q_{\max}$) is obtained by taking $\omega_{\rm s}=a_{\min}$ (resp. $\omega_{\rm s}=-a_{\min}$),
 \begin{align}\label{eq:q_bounds}
  q_{\min} &= \frac{n^-}{2}-\frac{1}{2\pi}\sum_{i=1}^n\arcsin\left(a_{\min}/|a_{i,i+1}|\right)\, , & q_{\max} &= \frac{n^-}{2}+\frac{1}{2\pi}\sum_{i=1}^n\arcsin\left(a_{\min}/|a_{i,i+1}|\right)\, .
 \end{align}
 
 When $\omega_{\rm s}=\pm a_{\min}$, for each edge such that $a_{i,i+1}=\pm a_{\min}$, the angle difference is $\pm\pi/2$. 
 Hence the corresponding term in the Jacobian matrix vanishes, and according to Lemma~\ref{lem:stability}, the corresponding fixed point is unstable. 
 Stable fixed points then cannot realize the extreme values of winding number $q_{\min}$ and $q_{\max}$ and the winding number belongs to the interval $(n^-/2-\sigma,n^-/2+\sigma)$. 
\end{proof}

We can even construct the stable synchronous states. 
Let $q\in(n^-/2-\sigma,n^-/2+\sigma)\cap\mathbb{Z}$ and assume that $\omega_{\rm s}$ solves~\eqref{eq:implicit_freq_stab}. 
The corresponding stable synchronous state can be constructed by defining $\bm{\theta}^{\omega_{\rm s}}$ as 
\begin{align}
 \theta_i^{\omega_{\rm s}} &= -n_i^- + \sum_{j=1}^{i-1}\arcsin\left(\omega_{\rm s}/|a_{j,j+1}|\right)\, , & i &\in \{1,...,n\}\, ,
\end{align}
where $n_i^-$ is the number of edges with negative weight on the path from vertex $1$ to vertex $i$. 
Furthermore this synchronous state is stable because by construction it satifies~\eqref{eq:stab_condi}.
 
Proposition~\ref{prop:freq}  implies that there is a finite number of stable synchronous states and it completely characterizes them. 
In particular, it allows to compute the number of stable synchronous states for identical positive weights and identical frequencies, which was obtained by Rogge and Aeyels~\cite{Rog04}.

\begin{corollary}[Rogge and Aeyels~\cite{Rog04}]
 For an oriented cycle with identical positive weights and identical frequencies, the number of stable synchronous states is ${\cal N} = 2\cdot{\rm Int}\left[(n-1)/4\right]+1$.
\end{corollary}

\begin{proof}
 Identical positive weights implies that $\sigma=n\pi/2$ and $n^-=0$. 
 Then $q_{\min}=-n/4$ and $q_{\max}=n/4$ and the number of integers in $(-n/4,n/4)$ is ${\cal N}=2\cdot{\rm Int}\left[(n-1)/4\right]+1$. 
\end{proof}

\begin{corollary}\label{cor:less_4}
 For cycles of length lower than or equal to $4$, there is a unique stable synchronous state. 
\end{corollary}

\begin{remark}
 As in section~\ref{sec:acyclic}, the results of this section are not subject to any limitation on the magnitude and sign of the coupling strenghts. 
 We thus cover all possible cases of oriented cycles with identical natural frequencies. 
\end{remark}

\subsection{Correlation between initial and final winding numbers}\label{sec:condi}
Theorem 3.4 of~\cite{Ha12} gives some explicit conditions for initial conditions to converge to a given synchronous state. 
We give here conditions on initial states limiting the possible final states. 
The knowledge of stable synchronous states and the LaSalle function $V$ in~\eqref{eq:lasalle} allow to derive some conditions on the final state of the system, for given initial conditions. 
We showed that $V$ is monotonically decreasing along the trajectories of the system. 
Thus, if $V(\bm{\theta}^\circ)<V(\bm{\theta}^*)$, starting at $\bm{\theta}^\circ$, the system cannot converge to the synchronous state $\bm{\theta}^*$. 

For instance, for positive identical coupling constants $a_{i,i+1}\equiv K$, it is known~\cite{Rog04} that the stable synchronous states satisfy $\theta_i-\theta_{i+1}=2\pi q/n$, 
where $q\in\mathbb{Z}$ is the winding number. 
At such a state, the function $V$ takes value $V_q\coloneqq 2nK[1-\cos(2\pi q/n)]$. 
Hence, for almost all initial conditions $\bm{\theta}^\circ$ such that $V(\bm{\theta}^\circ)\leq V_q$, the system will converge to a stable synchronous state with winding number $|q'|\leq |q|$. 
This limits the possible final states for given initial conditions (see Figure~\ref{fig:Vq}). 
\begin{figure}
 \centering
 \includegraphics[width=.8\textwidth]{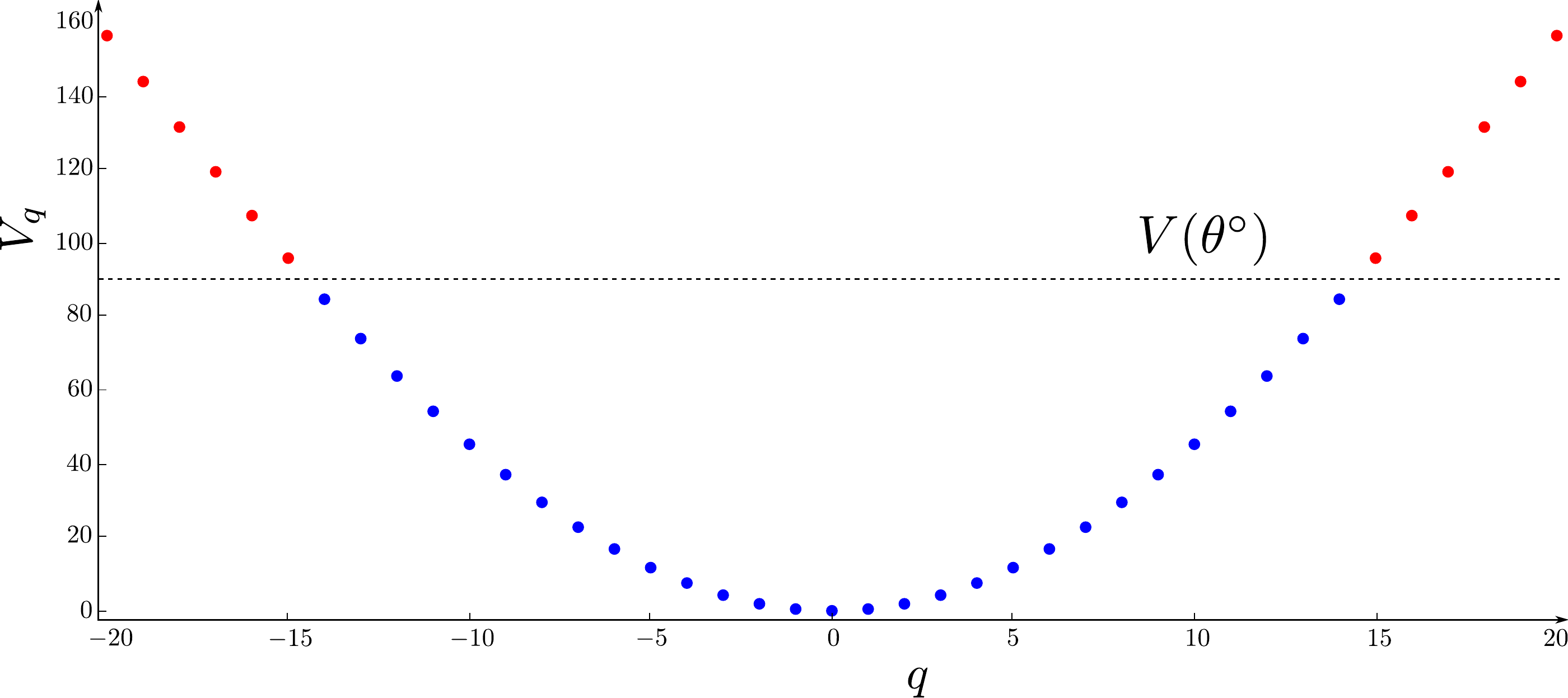}
 \caption{\it LaSalle function $V$ defined in~(\ref{eq:lasalle}) at the stable synchronous states with winding numbers $q\in\{-20,...,20\}$, 
 for an oriented cycle of length $n=83$ with coupling strength $a_{i,i+1}=1$ for all $i$. 
 If, at initial conditions $\bm{\theta}^\circ$, the function $V$ takes the value indicated by the dashed black line, then the final winding number corresponds to one of the blue dots, 
 and the red dots cannot be reached. }
 \label{fig:Vq}
\end{figure}

These observations suggest that along the time evolution of the system, the winding number tends to decrease. 
To corroborate this, we simulated the time evolution of \eqref{eq:main_cycle} for $100\,000$ random initial conditions, picked with uniform distribution on the state space. 
We consider an oriented cycle of $n=83$ oscillators, with identical natural frequencies, and for each simulation, we compare the initial and final winding numbers, 
$q_{\rm ini}$ and $q_{\rm fin}$ respectively. 
In the left panel of Figure~\ref{fig:distri_qf}, the cloud of points follows a line through the origin with slope larger than one. 
This implies that even if in some cases, the winding number increases along the simulation, there is a general trend towards a decrease of the winding number along the trajectory of the system, 
especially for large winding numbers. 
The center panel of Figure~\ref{fig:distri_qf} displays, for each values of $q_{\rm ini}$, the distribution of $q_{\rm fin}$ normalized to zero mean. 
To the eye, these curves look Gaussian. 
We then fit a Gaussian distribution to each of them and show their mean and standard deviation in the right panel of Figure~\ref{fig:distri_qf}. 
The slope of the mean $\mu$ is less than one, which indicates again that winding numbers tend to decrease (in absolute value) along simulations. 
It is striking to notice that the standard deviation (which was not normalized) is almost the same for all initial winding numbers. 
We have not been able to explain this occurence, which we did not expect. 

\begin{figure}
 \centering
 \includegraphics[width=\textwidth]{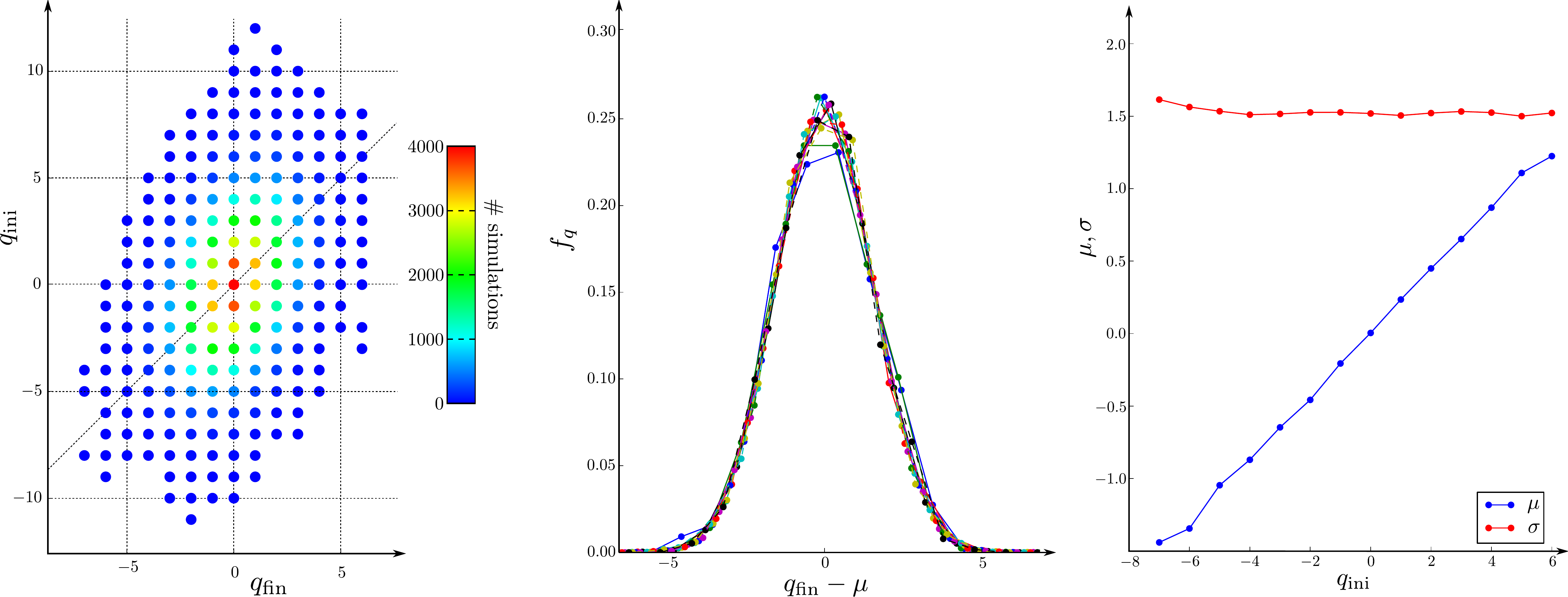}
 \caption{\it Left: Distribution of initial and final winding numbers for 100\,000 simulations of (\ref{eq:main_cycle}) with $n=83$, $\omega_i\equiv0$ and random initial conditions. 
 The color code indicates how many simulations started with winding number $q_{\rm ini}$ and ended at winding number $q_{\rm fin}$. 
 We see that the dynamics tends to reduce the winding number in absolute value.
 For instance, we see that $|q_{\rm fin}|<|q_{\rm ini}|$ for all simulations with $|q_{\rm ini}|>8$. 
 Center: Distribution of final winding number for our simulations. 
 Each curve corresponds to a different initial winding number $q_{\rm ini}$, ranging from $-7$ to $6$. 
 We consider only initial winding number with large enough statistics. 
 For each curve, the mean is normalized to zero. 
 To the eye, each curve follows a Gaussian distribution. 
 Notice that standard deviations are not normalized, i.e., each curve has approximately the same standard deviation.
 Right: Mean and standard deviation of the Gaussian fit for each curve of the center panel. 
 The slope of the mean with respect to $q_{\rm ini}$ is approximately $0.21$, corroborating the fact that the winding number tends to decrease along simulations. 
 We observe again that each curve has the same standard deviation. }
 \label{fig:distri_qf}
\end{figure}

\section{Combination of cyclic and acyclic networks}\label{sec:combi}
Based on the results of the previous sections, we extend next our results to networks built with oriented cycles and acyclic oriented graphs. 
Given an oriented graph ${\cal G}_1$, we say that a subgraph ${\cal G}_2$ is a \textit{leading component} if there are no edges going out of ${\cal G}_2$. 

\subsection{Cases with almost global synchronization}
\subsubsection{One large leading cycle}
Consider a graph composed of an acyclic oriented graph with a single leading component which is an oriented cycle (left panel of Figure~\ref{fig:leading_cycle}), 
whose oscillators have identical natural frequencies. 
The cycle synchronizes by Theorem~\ref{thm:cyc}, at a synchronization frequency $\omega_{\rm s}$.
The acyclic part is then led by oscillators with identical frequencies. 
Applying the same recursive argument as in the proof of Theorem~\ref{thm:acyc_suf}, we can show that all the acyclic part synchronizes for almost all initial conditions, 
provided that natural frequencies of its oscillators satisfy 
\begin{align}
 |\omega_i-\omega_{\rm s}| &\leq \rho_i\, , & i &\in \{1,...,n\}\, .
\end{align}

\subsubsection{Multiple small leading cycles}
Consider a graph composed of an acyclic oriented graph whose leading components are either simple leaders or a leading cycles of length at most $4$ (right panel of Figure~\ref{fig:leading_cycle}). 
Here the leaders and the oscillators of the leading cycles have identical natural frequencies. 
The conditions for global synchronization are then the same as in Theorem~\ref{thm:acyc_suf}. 
A cycle of length less than or equal to $4$ synchronizes to the frequency of its oscillators because it is too short to carry a nonzero winding number, see Corollary~\ref{cor:less_4}. 
The acyclic part is then led by oscillators with identical frequencies and almost always synchronizes following the same recursive argument as in the proof of Theorem~\ref{thm:acyc_suf}. 

\begin{figure}
 \centering
 \includegraphics[width=.9\textwidth]{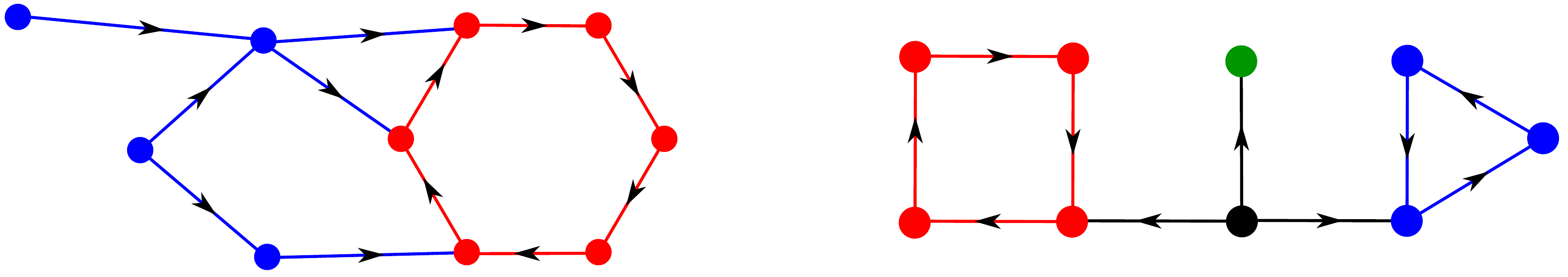}
 \caption{\it Left: Example of an oriented graph with only a leading cycle. 
  The leading cycle is in red and the the blue part is an acyclic oriented graph. 
  Such a network globally synchronizes.
  Right: Example of an almost globally synchronizing oriented graph with multiple leaders and leading cycles of length not larger than $4$.}
 \label{fig:leading_cycle}
\end{figure}

\subsection{Cases where global synchronization is not achieved}
\subsubsection{Multiple large leading cycle}
A graph composed of an acyclic oriented graph with at least two leading components, where (at least) one of them is a leading cycle of length 5 or more (left panel of Figure~\ref{fig:six_n_leader}). 
The leading cycle can synchronize to various frequencies. 
If it does not synchronizes to the same frequency as the other leader(s), then synchrony cannot be reached. 
But if it synchronizes to the same frequency, then the network synchronizes. 
Synchronization is then not global, but is possible for a set of initial conditions with nonzero measure, for some natural frequencies and edge weights.

\subsubsection{Two large cycles}
We numerically verified that two cycles of length more than 5 coupled as in the right panel of Figure~\ref{fig:six_n_leader} can synchronize or not depending on inital conditions.
We have not been able to get any analytical insight for this case. 

\begin{figure}
 \centering
 \includegraphics[width=.8\textwidth]{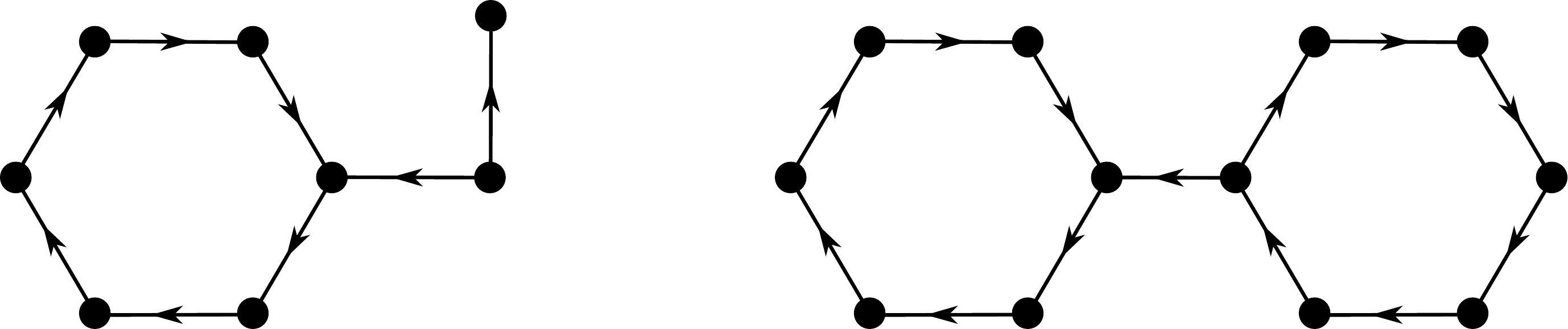}
 \caption{\it Left: Example of a nonglobally synchronizing network. 
  The cycle can synchronize to different frequencies depending on the initial conditions and the solitary leader has a fixed frequency.
  Right: Example of two cycles connected by an oriented edge. 
  Simulations indicate that this network does not always synchronize, depending on initial conditions. }
 \label{fig:six_n_leader}
\end{figure}

\section{Conclusion}
We studied synchronization in oriented and signed Kuramoto oscillator networks. 
By considering the simplest oriented interaction graphs, we completely characterized frequency-synchronization of the Kuramoto model on oriented and signed acyclic networks 
and on oriented and signed cycles with identical natural frequencies. 
All our results are valid regardless of the magnitude and sign of the couplings, they are then as general as possible. 

In oriented and signed acyclic networks, we gave necessary and sufficient conditions for almost global synchronization and we showed that, if it exists, there is a unique stable synchronous state. 
We further showed in section~\ref{sec:tight} that, in general, these conditions cannot be improved without more conditions on the networks considered. 

We proved that on oriented and signed cycles with identical natural frequencies, the Kuramoto oscillators always synchronizes. 
Furthermore, we showed that the number of stable synchronous states is finite, and we gave an explicit formula to compute their number. 
We observed in section~\ref{sec:condi} that, for initial conditions with a given winding number $q_{\rm ini}$, the distribution of winding numbers after the time evolution of~\eqref{eq:main_cycle} 
is Gaussian. 
It is striking that for all values of $q_{\rm ini}$, the variance of the distribution is identical. 
We are, at this point, not able to explain this fact. 

We finally showed through some examples, that in more general oriented interaction graphs, the dynamics become much more rich, 
even for some simple combinations of acyclic graphs and oriented cycles.

\appendix
\section{Asymptotically autonomous semiflows}\label{ap:mischaikow}
In~\cite{Mis95}, Mischaikow et al. relate the solution of the nonautonomous system $\dot{x}=f(t,x)$ on $\mathbb{R}^n$ to the solution of the autonomous system $\dot{y}=g(y)$ on $\mathbb{R}^n$, 
where $f(t,\cdot)\to g(\cdot)$ for $t\to\infty$. 
We give some preliminary definitions and state then two results that we need in the proof of Theorem~\ref{thm:acyc_suf}. 

\begin{definition}[Mischaikow et al.~\cite{Mis95}]\label{def:mischaikow}
 Let $\Theta\colon T\times \mathbb{R}^n\to \mathbb{R}^n$ be a continuous function, where $T=\{(t,s)\colon 0\leq s\leq t<\infty\}$. 
 The function $\Theta$ is a \emph{nonautonomous semiflow} on $\mathbb{R}^n$ if it satisfies 
 \begin{align}
  &\Theta(s,s,x) = x\, , && s \geq 0\, ;\\
  &\Theta(t,s,\Theta(s,r,x)) = \Theta(t,r,x)\, , && t \geq s\geq r\geq 0\, .
 \end{align}
 The semiflow is called \emph{autonomous} if in addition
 \begin{align}
  \Theta(t+r,s+r,x) &= \Theta(t,r,x)\, .
 \end{align}
 A nonautonomous semiflow $\Theta$ is called \emph{asymptotically autonomous with limit semiflow} $\Phi$, if $\Phi$ is an autonomous semiflow on $\mathbb{R}^n$ and 
 \begin{align}
  \Theta(t_j+s_j,s_j,x_j)\ &\to \Phi(t,x)\, , & j &\to \infty\, ,
 \end{align}
 for any sequences $t_j\to t$, $s_j\to\infty$ and $x_j\to x$, with $x,x_j\in \mathbb{R}^n$, $0\leq t,t_j<\infty$, and $s_j\geq0$. 
 
 The \emph{$\omega$-limit set} of $\Theta$ is 
 \begin{align}
  \Lambda_{\Theta}(s,x_0) &\coloneqq \left\{x\in \mathbb{R}^n\colon \exists~\{t_j\},~j\in\mathbb{N},~{\rm s.t.}~\lim_{j\to\infty}t_j=
  \infty~{\rm and}~\lim_{j\to\infty}\Theta(s,t_j,x_0)=x \right\}\, .
 \end{align}
\end{definition}

\begin{definition}[Mischaikow et al.~\cite{Mis95}]\label{def:mischaikow_inv}
 A subset $A$ of $\mathbb{R}^n$ is said to be \emph{positively invariant} for an autonomous semiflow $\Phi$ if for all $a\in A$ and $t\geq 0$, $\Phi(t,a)\in A$. 
 The subset $A$ is \emph{invariant} for $\Phi$ if it is positively invariant and for all $a\in A$ and $t\geq 0$, there exits $b\in a$ such that $\Phi(t,b)=a$. 
\end{definition}

\begin{definition}[Mischaikow et al.~\cite{Mis95}]\label{def:mischaikow_chain_rec}
 Let $A$ be a nonempty positively invariant subset of $\mathbb{R}^n$ and $x,y\in A$. 
 For $\varepsilon > 0$, $t > 0$, an $(\varepsilon,t)$\emph{-chain} from $x$ to $y$ (in $A$) is a sequence $\{x=x_1,x_2,...,x_{n+1}=y\, ; t_1,t_2,...,t_n\}$ of points 
 $x_i\in A$ and times $t_i\geq t$ such that $\|\Phi(t_i,x_i) - x_{i+1}\|<\varepsilon$, $i\in\{1,...,n\}$. 
 A point $x\in A$ is called \emph{chain recurrent} (in $A$) if for every $\varepsilon>0$, $t>0$ there is an $(\varepsilon,t)$-chain from $x$ to itself in $A$. 
 The set $A$ is said to be \emph{chain recurrent} if every point in $A$ is chain recurrent in $A$. 
\end{definition}

Roughly speaking, a set $A$ is chain recurrent if for any $x\in A$, there is a sequence of trajectories starting in $A$ remaining arbitrarily close to $A$, 
such that the end point of the $i^{\rm th}$ is arbitrarily close to the starting point of $(i+1)^{\rm th}$ trajectory, and such that the last trajectory ends arbitrarily close to $x$. 

Consider the systems of ordinary differential equations $\dot{x}=f(t,x)$ and $\dot{y}=g(y)$ on $\mathbb{R}^n$. 
Denote by $\Theta(t,s,x_0)$ the solution $x(t)$ of the first system, with $x(s)=x_0$, and denote by $\Phi(t,x_0)$ the solution $y(t)$ of the second system, with $y(0)=x_0$. 

\begin{proposition}[Mischaikow et al.~\cite{Mis95}]\label{prop:mischaikow}
 If $f(t,x)\to g(x)$, $t\to\infty$, uniformly\footnote{Here \emph{uniform} means that $\forall~\varepsilon>0$, 
 $\exists~T>0$ such that $|f(t,x)-g(x)|<\varepsilon$, $\forall~t\geq T$ and $\forall~x\in\mathbb{R}^n$.}
 on compact subsets of $\mathbb{R}^n$, 
 then $\Theta$ is asymptotically autonomous with limit semiflow $\Phi$. 
\end{proposition}

\begin{theorem}[Mischaikow et al.~\cite{Mis95}]\label{thm:mischaikow}
 Let $\Theta$ be an asymptotically autonomous semiflow with limit semiflow $\Phi$, and let its orbit $\{\Theta(t,s,x)\colon t\in[0,\infty)\}$ have compact closure in $\mathbb{R}^n$. 
 Then $\Lambda_{\Theta}(s,x_0)$ is invariant and chain recurrent for the semiflow $\Phi$ and attracts $\Theta(t,s,x_0)$.
\end{theorem}


\end{document}